\documentclass[journal]{IEEEtran}

\usepackage{times,amsmath,epsfig}
\usepackage{times}
\long\def\comment#1{}

\usepackage{pifont}
\usepackage{color}
\usepackage{amssymb}
\usepackage{graphicx}
\usepackage{xspace}
\usepackage{subfigure}
\usepackage{amsmath}
\usepackage{float}
\usepackage{amssymb}
\usepackage{latexsym}
\usepackage{graphicx}
\usepackage{url}
\usepackage{epsfig}
\usepackage{mathrsfs}
\usepackage{color}

\usepackage{algorithm}
\usepackage[noend]{algorithmic}

\newcommand{\kw}[1]{{\ensuremath {\mathsf{#1}}}\xspace}

\newcommand{\cs}{CS}
\newcommand{\e}{\kw{e}}
\newcommand{\head}{\kw{bp}}
\newcommand{\dist}{\kw{dist}}
\newcommand{\skyp}{SP}
\newcommand{\lbop}{\kw{LBOP}}
\newcommand{\findsp}{{\sc Optimal-Path}\xspace}
\newcommand{\computecstwo}{{\sc Compute-Cs-2D}\xspace}
\newcommand{\computecs}{{\sc Compute-Cs}\xspace}
\newcommand{\clbop}{{\sc Compute-LBOP}\xspace}
\newcommand{\proc}{{\sc Procedure}\xspace}
\newcommand{\vf}{{\sc Vertex-Filtering}\xspace}

\newcommand{\ct}{w}
\newcommand{\cost}{w}

\newcounter{definition}[section]
\renewcommand{\thedefinition}{\nthesection.\arabic{definition}}

\newenvironment{definition}{
     \refstepcounter{definition}
     {\vspace{1ex} \noindent\bf  Definition  \thedefinition:}}{
     \vspace{1ex}} 

\newcounter{theorem}[section]
\renewcommand{\thetheorem}{\nthesection.\arabic{theorem}}
\newenvironment{theorem}{\begin{em}
        \refstepcounter{theorem}
        {\vspace{1ex} \noindent\bf  Theorem  \thetheorem:}}{
        \end{em}\vspace{1ex}} 

\newcounter{lemma}[section]
\renewcommand{\thelemma}{\nthesection.\arabic{lemma}}
\newenvironment{lemma}{\begin{em}
        \refstepcounter{lemma}
        {\vspace{1ex}\noindent\bf Lemma \thelemma:}}{
        \end{em}\vspace{1ex}} 
        
\newcommand{\eop}{\hspace*{\fill}\mbox{$\Box$}}  
\newcommand{\nthesection}{\arabic{section}}
\newcommand{\stitle}[1]{\vspace{1ex} \noindent{\bf #1}}

\begin{document}
%
\title{An Efficient Index Method for the Optimal Route Query over Multi-Cost Networks}
%
%
%
\author{%
{Yajun Yang{\small $^{1}$}, Hang Zhang{\small $^{1}$}, Hong Gao{\small $^{2}$}, Qinghua Hu{\small $^{1}$}, Xin Wang{\small $^{1}$}}%
\vspace{1.6mm}\\
\fontsize{10}{10}\selectfont\rmfamily\itshape
$^{1}$College of Intelligence and Computing, Tianjin University, Tianjin, China\\
\fontsize{9}{9}\selectfont\ttfamily\upshape yjyang@tju.edu.cn, aronzhang@tju.edu.cn, huqinghua@tju.edu.cn, wangx@tju.edu.cn
\vspace{1.2mm}\\
\fontsize{10}{10}\selectfont\rmfamily\itshape
$^{2}$School of Computer Science and Technology, Harbin Institute of Technology, Harbin, China \\
\fontsize{9}{9}\selectfont\ttfamily\upshape
honggao@hit.edu.cn
}

\maketitle

\begin{abstract}
Smart city has been consider the wave of the future and the route recommendation in networks is a fundamental problem in it. Most existing approaches for the shortest route problem consider that there is only one kind of cost in networks. However, there always are several kinds of cost in networks and users prefer to select an optimal route under the global consideration of these kinds of cost. In this paper, we study the problem of finding the optimal route in the multi-cost networks. We prove this problem is NP-hard and the existing index techniques cannot be used to this problem. We propose a novel partition-based index with contour skyline techniques to find the optimal route. We  propose a vertex-filtering algorithm to facilitate the query processing. We conduct extensive experiments on six real-life networks and the experimental results show that our method has an improvement in efficiency by an order of magnitude compared to the previous heuristic algorithms.

\end{abstract}

\begin{IEEEkeywords}
optimal path, multi-cost networks, index
\end{IEEEkeywords}

\IEEEpeerreviewmaketitle

\section{introduction}

\IEEEPARstart{W}ith the rapid developing of the information technology, smart technologies have been widely used to promote the convenience for people's life in the city. Smart city has been attracting more and more attention from academic and industrial community. The intelligent route recommendation is a fundamental problem in smart city. For example, in  traffic networks, the shortest route query is to find a shortest path between two locations. In social networks, the shortest route query is to find the closest relationships such as friendship between two individuals. 

Most existing work about the shortest route problem assume that there is only one kind of cost in the networks. However, the relationships among
various entities are always investigated from several distinct aspects. For example, in traffic networks, the routes between two cities are taken into account with several kinds of cost such as road length, toll fee, traffic congestion and so on. It is inadvisable to choose a shortest path only by one kind of cost because the total toll fee of a route with the minimum length may be too expensive to accept for some users. It is important to find an optimal route under global consideration with people's preference. 

A network is called \emph{multi-cost network} if every edge in it has several kinds of cost. Obviously, the shortest route under one kind of cost may not be the optimal route for some users in multi-cost networks. Score function is proposed by user and it can calculate an overall score based on all kinds of cost to measure the optimality for a route. Note that the score functions given by distinct users may be different. Given a score function $f(\cdot)$, a starting vertex $v_s$ and an ending vertex $v_e$, this paper is to find a route from $v_s$ to $v_e$ with the minimum score and such route is also called an \emph{optimal path} from $v_s$ to $v_e$ under the score function $f(\cdot)$ in the following. 

The traditional shortest path problem can be solved by polynomial algorithm e.g., Dijkstra algorithm, and various index techniques are proposed to improve the efficiency. However, these index techniques cannot be used for the optimal path in the multi-cost networks because the score functions given by distinct users may be different. An index built for a score function $f(\cdot)$ cannot cope with the case of another score function $g(\cdot)$. In addition, we prove the optimal path problem is NP-hard in this paper if the score function is non-linear, e.g., $f(x,y)=x^2+y^2$, and then existing algorithms cannot work under such functions. As discussed in previous studies about traffic networks\cite{springerlink:10.1023/A:1009609820093,
tranportationcmshetty}, the non-linear score
functions are existent widely and reasonable in real-life. For example, in special conditions such as traffic jam occurring, the traveling time and fuel consumption are nonlinear (e.g., quadratic, convex and so on) function with the distance from source to destination\cite{mokhtarbook}.

In this paper, we develop a novel partition-based index to find the optimal path in multi-cost networks under various linear or non-linear score functions. The main contributions are summarized below. First, we study the problem of the optimal path recommendation in multi-cost networks and prove it is NP-hard. Second, we propose a partition-based index and contour skyline in the index. We prove the problem of computing contour skyline is NP-hard. We give a $2$-approximate algorithm and present that there is no $(2-\epsilon)$-approximate solution in polynomial time if $P \neq NP$. Third, we propose a vertex-filtering algorithm which can filter a
large of proportion of vertices that cannot be passed through by the optimal path. Finally, we confirm the effectiveness and efficiency
of our algorithms using real-life datasets.

The rest of this paper is organized as follows. Section
\ref{prostate} gives the problem statement. 
Section \ref{main-index} introduces the partition-based index and how to
construct it. Section \ref{main-query} proposes a vertex-filtering
algorithm and discusses how to find the optimal path by
partition-based index. We conduct experiments using six real-life
datasets in Section \ref{performance}. The experimental results confirm the effectiveness and
efficiency of our approach. Section \ref{sec-related} discusses the
related works. We conclude this paper in section \ref{conc}.

\section{Problem Statement} \label{prostate}

\subsection{Multi-cost Networks and the Optimal Path}

\begin{definition}
({\bf multi-cost network})
A multi-cost network is a simple directed graph, denoted as $G=(V,E,W)$,
where $V$ and $E$ are the sets of vertices and edges respectively. $W$ is a set
of vectors. Every edge $e \in E$ is represented by $e=(v_i,v_j)$, $v_i,v_j \in V$, and $w(v_i,v_j) \in W$ is the cost vector of $(v_i,v_j)$, $w(v_i,v_j)=(w_1,w_2,\cdots,w_d)$, where $w_i$ is the $i$-th kind of cost value
of edge $(v_i,v_j)$.
\end{definition}

In this paper, we assume $w_i \geq 0$. This assumption
is reasonable, because the cost cannot be less than zero in real
applications. Our work can be easily extended to handle undirected
graphs, an undirected edge is equivalent to two directed
edges. For simplicity, we only discuss the
directed graphs in the following.

A path $p$ is a sequence of vertices $(v_0, v_1, \cdots, v_l)$,
where $v_i \in V$ and $(v_{i-1},v_i)\in E$ 
We use $\ct(p)$ to denote cost vector of path $p$, i.e., $w(p)=(w_1(p),w_2(p),\cdots,w_d(p))$, where $w_x(p)=\sum_{i=1}^{l}w_x(v_{i-1},v_i)$ for $0\leq x \leq d$.



For a path $p$ in $G$, a score function is used to calculate an overall score $f(p)$ base on $w(p)$. The score function $f(\cdot)$ is always monotone increasing, i.e., for two
different paths $p$ and $p'$, if $(\forall i, c_i(p)\leq c_i(p'))
\wedge (\exists i, c_i(p)<c_i(p'))$, then $f(p)<f(p')$. It is a common propertyand its intuitive meaning is that if all
costs of a path $p$ are less than that of $p'$, then
the overall score of $p$ must be less than $p'$. The definition of the
optimal path over the multi-cost networks is given below:



\begin{definition}
({\bf optimal path}) Given a multi-cost network $G$, a score
function $f(\cdot)$, a starting vertex $v_s$ and an ending vertex $v_e$, the optimal path from $v_s$ to $v_e$, denoted as $p^*_{s,e}$, is a
path in $G$ that has the minimum score among all paths from $v_s$ to
$v_e$, i.e., $f(p^*_{s,e})\leq f(p)$ for any $p \in P_{s,e}$, where $P_{s,e}$ is the set of all simple paths from $v_s$ to
$v_e$.
\end{definition}

Fig.~\ref{fig1} illustrates an concrete multi-cost network $G$. The score function in this example is $f(w_1,w_2)=w_1+w_2$. Consider the path $p:
v_s\rightarrow v_1 \rightarrow v_e$ in $G$, its cost vector is $\cost(p)=(10,4)$
and its score is $f(p)=w_1(p)+w_2(p)=10+4=14$. because the score of $p$ is the minimum among all
paths from $v_s$ to $v_e$, then $p$ is the optimal path.

The following theorem shows the problem of finding the optimal path in the multi-cost networks under non-linear score function is NP-hard.

%

%

\begin{theorem}\label{theorem-procomplex}
The problem of finding the optimal path under a non-linear function in the multi-cost networks is NP-hard.
\end{theorem}

\begin{proof} We reduce the problem of the minimum sum of squares, which is NP-complete\cite{DBLP:books/fm/GareyJ79}, to this problem. The minimum sum of squares problem is as follows. Given a number set $A=\{a_1,a_2,\cdots,a_n\}$ of size $n$ and an integer $k \leq |A|$, find a partition ${\cal A}^*=\{A_1, A_2, \cdots, A_k\}$ of $A$ such that $\sum_{j=1}^k(\sum_{a_i\in A_j}a_i)^2$ is minimum. Note that $A_j~(1\leq j \leq k)$ cannot be an empty set for an optimal partition ${\cal A}^*$. Given an instance of the minimum sum of squares problem, it can be converted to an instance of the optimal path problem as follows. We create a graph $G$ with $n+1+kn$ vertices, $\{v_1, v_2, \cdots, v_{n+1}\} \cup \{v_{i,j}|1 \leq i \leq n, 1\leq j \leq k\}$. Here, $v_{i,j}(1\leq j \leq k)$ is placed between $v_i$ and $v_{i+1}$. We create the edges in $G$ as follows. For $\forall 1\leq i \leq n$ and $\forall 1\leq j \leq k$, we create an edge $e_{i,(i,j)}$ from $v_i$ to $v_{i,j}$. The cost of edge $e_{i,(i,j)}$ is assigned as $\ct(e_{i,(i,j)})=(0,\cdots,0,\frac{a_i}{2},0,\cdots,0)$, i.e., the $j$-th cost value of $\ct(e_{i,(i,j)})$ is $\frac{a_i}{2}$ and the others are zero. Similarly, we create an edge $e_{(i,j),i+1}$ from $v_{i,j}$ to $v_{i+1}$. The cost of edge  $e_{(i,j),i+1}$ is also $\ct(e_{(i,j),i+1})=(0,\cdots,0,\frac{a_i}{2},0,\cdots,0)$, i.e., the $j$-th cost value of $\ct(e_{(i,j),i+1})$ is $\frac{a_i}{2}$ and the others are zero. Let $v_1=v_s$ and $v_{n+1}=v_e$. Score function is $f(w_1,\cdots,w_k)=\sum_{i=1}^k(w_i)^2$. Here, $(w_1,\cdots,w_k)$ is the cost vector $\ct(p)$ of a path $p$. Obviously, if a path $p$ travels through an edge $e_{i,(i,j)}$, it must travel through $e_{(i,j),i+1}$. We can concatenate $e_{i,(i,j)}$ and $e_{(i,j),i+1}$ as a new edge $e_{i,i+1}^j$ from $v_i$ to $v_{i+1}$. $e_{i,i+1}^j$ is called the $j$-th edge from $v_i$ to $v_{i+1}$ in $G$. The cost of $e_{i,i+1}^j$ is $(0,\cdots,0,a_i,0,\cdots,0)$, i.e., the $j$-th cost value of $\ct(e_{i,i+1}^j)$ is $a_i$ and the others are zero. For any path $p$ from $v_s$ to $v_e$ in graph $G$, the $j$-th cost value $w_j(p)$ of $\ct(p)$ is equal to the sum of the $j$-th cost values of all the edges in $p$. Let $E^j_p$ be the set of all the $j$-th edges in $G$ that $p$ travels through, i.e., $E^j_p=\{e_{i,i+1}^j|e_{i,i+1}^j\in p, 1\leq i \leq n\}$.  Then $\{E^j_p|1\leq j \leq k\}$ corresponds to a partition ${\cal A}=\{A_j|1\leq j \leq k\}$ of $A$, where $A$ is the number set $\{a_1,a_2,\cdots, a_n\}$ and $A_j~(1\leq j \leq k)$ is the number set of the $j$-th cost value of all the edges in $E^j_p$, i.e., $A_j=\{w_j(e)|e\in E^j_p\}$. Consequently, an optimal path $p^*$ with the minimum score corresponds to an optimal partition ${\cal A}^*$ for $A$ such that $\sum_{j=1}^k(\sum_{a_i\in A_j}a_i)^2$ is the minimum. Note that this reduction is in polynomial time. If we find an optimal path from $v_s$ to $v_e$ in $G$ in polynomial time, then we also can find an optimal partition ${\cal A}^*$ for number set $A$. Therefore, the problem of finding the optimal path over the multi-cost graphs is NP-hard. \eop
\end{proof}


\begin{figure}[t]
\centering
\includegraphics[width=1.6in]{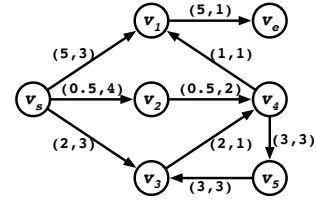}
\caption{An example of multi-cost graph $G(V,E)$} \label{fig1}
\end{figure}

\subsection{Challenging Problem} \label{sec-challenging-problem}



If score function $f(\cdot)$ is linear, i.e., for any two consecutive edges
$(v_x,v_y)$ and $(v_y,v_z)$, we have
\begin{displaymath}
f(w(v_x,v_y)+w(v_y,v_z))=f(w(v_x,v_y))+f(w(v_y,v_z))
\end{displaymath}
then $f(w(v_x,v_y))$ can be considered as the single-one weight of the edge $(v_x,v_y)$ for any edge in $G$. Obviously, $f(w_1,w_2)=w_1+w_2$ is a linear function.
In this case, the problem of finding the optimal path in the multi-cost networks can be solved in polynomial time by the existing shortest path algorithms, e.g., Dijkstra algorithm. The shortest path $p$ based on the weight $f(w(v_x,v_y))$
is exactly the optimal in the multi-cost networks. Otherwise, there is
another path $p'$ such that $f(p')<f(p)$. By the linearity of score
function, we have
\begin{align*}
&f(p')=f(\sum_{i=1}^{l-1}w(v'_{i},v'_{i+1}))=\sum_{i=1}^{l-1}f(w(v'_{i},v'_{i+1}))\\
<&f(p)=f(\sum_{i=1}^{r}w(v_i,v_{i+1}))=\sum_{i=1}^{r}f(w(v_1,v_{i+1}))
\end{align*}
which is in contradiction to the correctness of Dijkstra algorithm.
Most existing works on the shortest path problem propose various index techniques to improve the efficiency. However, the existing index techniques cannot be used for this problem even though the score function is linear. The reason is the score functions given by distinct users may be different. An index built for a score function $f(\cdot)$ cannot cope with the case of another score function $g(\cdot)$.


If score function $f(\cdot)$ is non-linear, that is, 
\begin{displaymath}
f(w(v_x,v_y)+w(v_y,v_z)) \neq f(w(v_x,v_y))+f(w(v_y,v_z))
\end{displaymath}
then the optimal path problem in the multi-cost networks cannot be solved by existing methods for traditional shortest path problem. Most of these methods are based on the following property: any sub-path of a shortest path is also a shortest path. They maintain the shortest paths for some pairs of vertices in an index and answer the query by concatenating the shortest paths to be visited inside index and outside index.
%
%
However, the property of the optimal sub-path is not correct for the multi-cost graphs when the score function is non-linear. Consider the example in
Fig.~\ref{fig1}, if the score function is set as $f(w_1,w_2)=w_1^2+w_2^2$, which is monotonically increasing in the region of $\{x\geq 0, y\geq 0\}$, then the optimal path from $v_s$ to $v_5$ is $v_s\rightarrow v_2\rightarrow v_4\rightarrow v_5$. Note that the sub-path $p:s\rightarrow v_2\rightarrow v_4$
is not the optimal path from $v_s$ to $v_4$, because its score is
$f(1,6)=37$, which is less than the score $f(4,4)=32$ of path
$p':s\rightarrow v_3\rightarrow v_4$. This example states a sub-path of an
optimal path may be not the optimal one in the multi-cost networks.

Enumeration is a straightforward method to compute the optimal path in the multi-cost graphs. 
Given a starting vertex $v_s$ and an ending vertex $v_e$, we compute the score for every path
from $v_s$ to $v_e$ and then find the path with the minimum score. Let
the maximum out-degree of $G$ is $\lambda$, i.e.,
$\lambda=max\{d^+(v)|v\in V\}$, where $d^+(v)$ is out-degree of $v$.
The search space is $O(\lambda^{|V|})$ for enumeration, which is obviously infeasible in real applications. Another alternative approach is to
pre-compute the optimal path for every pair of vertices in $G$. The
critical shortcoming is that cannot cope with
distinct score functions. Since the score functions are various, an
optimal path under one function may be not an optimal path under another function.

There are only a small number of heuristic algorithms are proposed to solve it\cite{DBLP:conf/cikm/YangYGL12}. In this paper, we develop a novel partition-based index to find the optimal path in multi-cost networks and it can support well for Dijkstra-based algorithms under linear functions or heuristic algorithms under non-linear functions. 



%
%
\section{Partition-Based Index} \label{main-index}

\subsection{What is the Partition-Based Index?}\label{subsec-whatk}

Given a graph $G(V,E)$, a $k$-partition of $G$ is a collection
$\{V_1,\cdots,V_k\}$ satisfying the following conditions: (1) every $V_p$ is a subset of $V$;
(2) for $\forall V_p, V_q$ $(p \neq q)$, $V_p\cap V_q=\emptyset$; (2)$V=\bigcup_{1\leq p \leq
k}V_p$. A vertex $v_i$ is called an \emph{\textbf{entry (or exit)}} of $V_p$, if (1) $v_i\in
V_p$; and (2) $\exists v_j$, $v_j\notin V_p\wedge v_j\in N^{-}(v_i)~(\text{or}~v_j\in N^{+}(v_i))$, where $N^{-}(v_i)$ and $N^{+}(v_i)$ are $v_i$'s incoming and
outgoing neighbor set respectively. Entries and exits are also called the \emph{\textbf{border vertices}}. We use $V_{p}.entry$ and $V_p.exit$ to denote the entry set
and exit set of $V_p$, and use $V.entry$ and $V.exit$ to denote the sets of all entries and exits in $G$, respectively. Obviously, $V.entry = \bigcup_{1\leq
p\leq k}V_p.entry$ and $V.exit = \bigcup_{1\leq p\leq k}V_p.exit$.


A partition-based index includes two parts: \emph{\textbf{inter-index}} and
\emph{\textbf{inner-index}}. We first introduce the \textit{lower bound of optimal path} (\lbop) and \textit{skyline path}.

For a multi-cost network $G$ with $d$ kinds of cost, $\mathcal{G}_x$~$(1\leq x \leq d)$ is a weighted graph with the same structure as $G$, and the weight of every edge $(v_i,v_j)$ in $\mathcal{G}_x$ is the $x$-th cost $w_x(v_i,v_j)$ of $w(v_i,v_j)$. For any two vertices $v_i,v_j \in G$, $\mathcal{P}_{i,j}=\{p^1_{i,j},\cdots,p^d_{i,j}\}$
is \textit{the set of single-one cost shortest paths} from $v_i$ to $v_j$, where $p^x_{i,j}$ is the shortest path from $v_i$ to $v_j$ in $\mathcal{G}_x$. We use $\phi^x_{i,j}$ to denote the weight of $p^x_{i,j}$. The cost vector
$\Phi_{i,j}=(\phi^1_{i,j},\cdots,\phi^d_{i,j})$ is called the
\textit{\textbf{lower bound of the optimal path}} (\lbop) from $v_i$ to $v_j$ in $G$.

Let $p$ and $p'$ be two different paths in a multi-cost graph $G$. We say $p$
dominate $p'$, denoted as $p\prec p'$, iff for $\forall i$ $(1\leq
i\leq d)$, $w_i(p)\leq w_i(p')$, and $\exists i$ $(1\leq i\leq d)$,
$w_i(p)<w_i(p')$. Here, $w_i(p)$ and $w_i(p')$ are the $i$-th cost
value of $\cost(p)$ and $\cost(p')$, respectively. For two vertices $v_i, v_j \in G$, a path $p$ is a \textit{\textbf{skyline path}} from $v_i$
to $v_j$ iff $p$ cannot be dominated by any other path $p'$ from $v_i$ to $v_j$.

For any path $p_{i,j}$ from $v_i$ to $v_j$, the cost vector of $p_{i,j}$ is
$\cost(p_{i,j})=(w_1(p_{i,j}),\cdots,w_d(p_{i,j}))$, then we have
$\Phi_{i,j}\preccurlyeq p_{i,j}$, i.e., for $\forall x$ $(1\leq x\leq d)$,
$\phi^x_{i,j} \leq w_x(p_{i,j})$.

Lemma \ref{lemma3} guarantees that $\Phi_{i,j}$ is the strict lower bound for
the optimal path from $v_i$ to $v_j$ in the multi-cost network $G$.  

\begin{lemma}\label{lemma3}
$\Phi_{i,j}$ is the strict lower bound for the optimal path from $v_i$
to $v_j$ in $G$, that is, there does not exist another lower bound
$\Phi'_{i,j}$ such that $\Phi_{i,j}\prec \Phi'_{i,j}$ and
$\Phi'_{i,j}\preccurlyeq p_{i,j}$ for any path $p_{i,j}$ from $v_i$ to $v_j$.
\end{lemma}

\begin{proof}
We prove it by contradiction. Assume that there is
$\Phi'_{i,j}$ satisfying $\Phi_{i,j}\prec \Phi'_{i,j}$, then
$\exists x$ $(1\leq x \leq d)$, such that
$\phi'^x_{i,j}>\phi^x_{i,j}$. On the other hand, because
$p^x_{i,j}$ is a path from $v_i$ to $v_j$ and then $\Phi'_{i,j}\preccurlyeq
p^x_{i,j}$. It means $\phi'^x_{i,j}\leq \phi^x_{i,j}$, which is a contradiction. \eop
\end{proof}


{\bf Inter-index:} Inter-index is essentially a matrix $A$ to maintain the \lbop for every pair of
border vertex and entry in $G$. Each row represents a border vertex (entry or exit) $v_i$ and each column
represents an entry $v_j$ in $G$. The size of $A$ is  $(|V.exit|+|V.entry|)\times |V.entry|$. Each cell $A_{i,j}$ includes two
elements: $\Phi_{i,j}$ and $\mathcal{P}_{i,j}$.


{\bf Inner-index:} Inner-index consists of $k$ sub-indexs and every
sub-index $I_p$ is associated with a vertex subset $V_p$. $I_p$
includes two parts: (i) \emph{Skyline-Path-Inner-Index} $I_{p}^{S}$;
and (ii) \emph{\lbop-Inner-Index} $I_{p}^{L}$.

Skyline-Path-Inner-Index $I_{p}^{S}$ of $V_p$ is a collection of skyline path sets for all pairs of entry and exit in $V_{p}$, i.e.,
$I_{p}^{S}=\{\skyp_{(i,j);p}|v_i\in V_p.entry, v_j\in V_p.exit\}$.
$\skyp_{(i,j);p}$ is the set of all skyline paths from $v_i$ to $v_j$ in $G_p$,
where $G_p$ is the induced subgraph of $V_p$ on $G$. Note that the paths in $\skyp_{(i,j);p}$ only pass through the
vertices in $V_p$.

\lbop-Inner-Index $I_{p}^{L}$ of $V_p$ is essentially a matrix $M_{p}$ of size $|V_{p}|\times |V_{p}|$ to maintain $\lbop$s for all pairs of vertices $v_i$ and $v_j$ $V_{p}$. 
Actually, we only need to maintain a smaller matrix $M'_{p}$ as $I_{p}^{L}$ in memory. $M'_{p}$ is a sub-matrix of $M_{p}$. It maintain all the $\lbop$s from an entry to a vertex in $V_p$ and all the $\lbop$s from a vertex to an exit in $V_p$. The remaining sub-matrix $M_p^{-}=M_{p}\setminus M'_p~(1\leq p \leq k)$ is maintained in the disk. $M_s^{-}$ and $M_e^{-}$ are taken into the memory when the starting vertex $v_s$ and the ending vertex $v_e$ are given.


By inter-index and \lbop-inner-index, $\Phi_{i,j}$ can be calculated easily
for any pair of vertices $v_i$ and $v_j$ in $G$. Given a starting vertex $v_s$
and an ending vertex $v_e$, we use $V_s$ and $V_e$ to denote the vertex subsets including
$v_s$ and $v_e$ respectively. If $V_s = V_e$, we can
obtain $\Phi_{s,e}$ from \lbop-inner-index $I_{p}^{L}$ directly.
If $V_s \neq V_e$, we calculate $\Phi_{s,e}$ by Lemma \ref{lemma6}.

\begin{lemma}\label{lemma6}
Given two vertices $v_s$ and $v_e$ in a multi-cost network $G$, $V_s$
and $V_e$ are two distinct vertex subsets including $v_s$ and $v_e$ respectively. Let $v_i$ be an entry of $V_e$. Thus for $\forall x$ $(1\leq x \leq d)$, we have
$\phi_{s,e}^x=\min\{\phi_{s,i}^x+\phi_{i,e}^x|v_i\in V_e.entry\}$, where $\phi_{s,e}^x$, $\phi_{s,i}^x$ and $\phi_{i,e}^x$ are the $x$-th cost of \lbop $\Phi_{s,e}$, $\Phi_{s,i}$ and $\Phi_{i,e}$ respectively.
%
%
\end{lemma}


\begin{proof}
We know $\phi_{(s,e);x}$ $(1\leq x\leq
d)$ is the weight of the shortest path $p_{s,e}^x$ in
graph $\mathcal{G}_x$, which must pass
through an entry $v_i$ in $V_e.entry$. Therefore, $p_{s,e}^x$ can be regarded as
two parts: (i) sub-path from $v_s$ to $v_i$; and (ii) sub-path
from $v_i$ to $v_e$. Because $\phi_{(s,i);x}$ and $\phi_{(i,e);x}$ are the weights
of the shortest paths from $v_s$ to $v_i$ and from $v_i$ to $v_e$ respectively
in $\mathcal{G}_x$, then we have $\phi_{(s,i);x}+\phi_{(i,e);x}\leq
\phi_{(s,e);x}$. On the other hand, $\phi_{(s,e);x}$ is the minimum
among all the paths from $v_s$ to $v_e$, then $\phi_{(s,e);x} \leq
\phi_{(s,i);x}+\phi_{(i,e);x}$. Thus we have
$\phi_{(s,e);x}=\phi_{(s,i);x}+\phi_{(i,e);x}$. Next, we prove that $v_i$
is exactly the entry minimizing $\phi_{(s,i);x}+\phi_{(i,e);x}$. It
is obvious otherwise $p_{s,e}^x$ is not the single-one
cost shortest path in $\mathcal{G}_x$. Then we have
$\phi_{(s,e);x}=\min\{\phi_{(s,i);x}+\phi_{(i,e);x}|v_i\in
V_e.entry\}$. \eop
\end{proof}

$\Phi_{s,e}$ can be calculated in two cases: (1) $v_s \in V_s.entry \cup
V_s.exit$; and (2) $v_s \notin V_s.entry \cup V_s.exit$. For case
(1), $\phi_{s,i}^x$ and $\phi_{s,i}^x$ can be directly retrieved from inter-index and \lbop-inner-index $I_e^L$ respectively. Therefore, the minimum
value of $\phi_{(s,i);x}+\phi_{(i,e);x}$ can be easily calculated as $\phi_{s,e}^x$ by Lemma \ref{lemma6}. For case (2), because $\phi_{s,i}^x$ is not maintained in 
inter-index, it is necessary to calculate the minimum value of $\phi_{s,j}^x+\phi_{j,i}^x|v_j\in V_s.exit\}$ as $\phi_{s,i}^x$ and then calculate $\phi_{s,e}^x$ in the similar way as the case (1).
%
%
The algorithm to compute $\Phi_{s,e}$ for any two vertices $v_s$ and
$v_e$ in $G$ is shown in Algorithm \ref{alg2}. The set
$\mathcal{P}_{s,e}$ of the single-one cost shortest paths can be calculated
in the similar way as calculating $\Phi_{s,e}$.

\begin{algorithm}[t]
  \caption{\clbop($I,s,t$)}
  \label{alg2}

{\small
\begin{tabbing}
{\bf Input:} \hspace{0.2cm}\= index $I$, starting vertex $v_s$ and ending vertex $v_e$\\
{\bf Output:} \> \lbop $\Phi_{s,e}$ from $v_s$ to $v_e$.
\end{tabbing}

\begin{algorithmic}[1]
 \IF {$V_s = V_e$}
    \STATE {\bf return} $\Phi_{s,e}$ from $I_{s}^{L}$(or $(I_e^{L})$);
 \ELSE
    \IF {$v_s \in V_s.entry \cup V_s.exit$}
       \STATE \proc$(v_s,v_e,V_e.entry)$;
    \ELSE
       \FOR {$v_i \in V_e.entry$}
          \STATE \proc$(v_s,v_i,V_s.exit)$;
       \ENDFOR
       \STATE \proc$(v_s,v_e,V_e.entry)$;
    \ENDIF
    \STATE {\bf return} $\Phi_{s,e}$;
 \ENDIF
 \end{algorithmic}
}
\end{algorithm}

\begin{algorithm}[t]
  \caption{\proc($v_i,v_j,V$)} \label{alg3}
{\small
\begin{algorithmic}[1]
\FOR {$x=1$ to $d$}
          \FOR {each $v_r \in V$}
             \STATE $\phi^{*}\leftarrow \phi_{(i,r);x}+\phi_{(r,j);x}$;
             \IF {$\phi_{(i,j);x}>\phi^{*}$}
                \STATE  $\phi_{(i,j);x} \leftarrow \phi^{*}$;
             \ENDIF
          \ENDFOR
    \ENDFOR
 \end{algorithmic}
}
\end{algorithm}

\subsection{How to Construct Partition-Based Index?} \label{subsec-sky}

\subsubsection{Inter-index and \lbop-inner-index}

\comment{ As we discussed above, $k$-cluster index includes
inter-index and inner-index and inner-index $I_j$ for each cluster
$V_i$ also includes two components: \lbop inner-index $I_j^{L}$ and
skyline path inner-index $I_{j}^{S}$.}

For \lbop-inner-index $I_p^{L}$ of vertex subset $V_p$, the shortest path algorithms can be used to calculate $\Phi_{i,j}$ for every pair of vertex $v_i$ and $v_j$ in $V_p$. For inter-index, $\Phi_{i,j}$ for every pair of border vertex $v_i \in V.entry \cup V.exit$ and entry $v_j\in V.entry$ also can be calculated by the shortest path algorithms.
It worth noting that it is not necessary to maintain $\Phi_{i,j}$ in inter-index if $v_i$ and $v_j$ are in the same vertex subset $V_p$ because it has been maintained in the \lbop-inner-index.


\subsubsection{Skyline-path-inner-index}

For every $I_p^S$ in Skyline-path-inner-index, $I_{p}^{S}=\{\skyp_{(i,j);p}|v_i\in V_p.entry, v_j\in
V_p.exit\}$, it is necessary to calculate $\skyp_{(i,j);p}$ for every pair of entry $v_i$ and exit $v_j$ in $V_p$. 
We use the heuristic algorithm proposed in \cite{DBLP:conf/cikm/YangYGL12} to calculate $\skyp_{(i,j);p}$. All possible skyline paths in $G_p$ are organized in a search tree $T$ and a prior queue $Q$ is used to maintain the paths in $T$ to be searched, where $G_p$ is the induced subgraph of $V_p$ on $G$. In each iteration, a path $p$ is dequeued from $Q$. When the ending vertex of $p$ is not $v_j$, algorithm need to check whether $p$ can be dominated by a path in $\skyp_{(i,j);p}$. If not, $p$ is extended to a new path $p'$ by appending an outgoing neighbor $v_o$ of ending vertex in $p$ and then $p'$ is inserted into $Q$. When the ending vertex of $p$ is $v_j$. If $p$ cannot be dominated by any path in $\skyp_{(i,j);p}$, $p$ will be inserted into $\skyp_{(i,j);p}$. On the other hand, the paths dominated by $p$ will be removed from $\skyp_{(i,j);p}$. The several pruning strategies can be used for this algorithm and the more details are shown in \cite{DBLP:conf/cikm/YangYGL12}.

\subsection{Contour skyline set}

Given a skyline-path-inner-index $I_p^{S}$, each skyline path $p\in \skyp_{(i,j);p}$ can be regarded as a skyline point $p$ in the
$d$-dimensional space according to $\cost(p)$. Note that some such points in the space are proximity. This property is helpful for improve
the efficiency of the optimal path query. In this section, we propose the
definition of the contour skyline set. All skyline points in $\skyp_{(i,j);p}$ can
be partitioned into several groups by their space proximity.
We compute a \textit{\textbf{contour skyline point}} for every group and the set of the
contour skyline points is called the \textit{\textbf{contour skyline set}} of
$\skyp_{(i,j);p}$.

Fig.~\ref{fig2} is an example of the contour skyline set in the cluster $V_p$. $p_1,\cdots,p_9$ are the skyline points
in a 2-dimensional space and each $p_i$ is a skyline path $p_i$. We observe that $R_1=\{p_1,p_2,p_3\}$,
$R_2=\{p_4,p_5,p_6,p_7\}$ and $R_3=\{p_8,p_9\}$ are three groups
such that the skyline points in the same group are space proximity.
Then $cp_1$, $cp_2$ and $cp_3$ are the contour skyline points
corresponding to $R_1$, $R_2$ and $R_3$ respectively. Let
$\cost(cp_i)=(w_1(cp_i),w_2(cp_i))$ be the cost vector of $cp_i$. It
is obvious that $cp_i$ is the \lbop of the skyline paths in $R_i$,
i.e., $w_x(cp_i)=\min\{w_x(p)|p\in R_i\}$, where $w_x(cp_i)$ and
$w_x(p)$ are the $x$-th cost value of $\cost(cp_i)$ and $\cost(p)$ respectively.
Therefore, the problem to compute the contour skyline points is equivalent to 
partition the skyline points into several different groups such that the points in each group are more space
proximity. Given a specified $r$, our goal is to partition the
skyline points into $r$ groups. To do that, we introduce the concept
of the diameter for such group. For a group $R_i$, the diameter of
$R_i$, denoted as $\mathcal{D}(R_i)$, is defined as the maximum
Euclidean distance among all the pairs of the points in $S$. Formally,
\begin{equation}\label{eq1}
\mathcal{D}(R_i)=max\{\dist(p,p')|p_i,p_j\in R_i\}
\end{equation}
where, $\dist(p,p')$ is the Euclidean distance between $p$ and
$p'$ in the multi-dimensional space. Given a $r$-partition
$\mathcal{R}=\{R_1,\cdots,R_r\}$, we define the diameter
$\mathcal{D}(\mathcal{R})$ of $\mathcal{R}$ below:
\begin{equation}\label{eq2}
\mathcal{D}(\mathcal{R})=max\{\mathcal{D}(R_i)|R_i \in \mathcal{R}\}
\end{equation}
Intuitively, $\mathcal{D}(\mathcal{R})$ quantifies the partition
quality as the maximum distance between any two points in the same
group. A partition $\mathcal{R}$ is good if, for every two points in
the same group, they are close to each other.


\begin{figure} \label{fig2}
\centering
\includegraphics[width=1.2in]{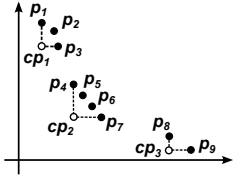}
\caption{An example of contour skyline set}
\label{fig2}
\end{figure}

\begin{definition}
({\bf Contour skyline}) Given two vertices $v_x$ and $v_y$ in vertex subset $V_p$, $\skyp_{(x,y);p}$ is the skyline path set
from $v_x$ to $v_y$ in the induced subgraph $G_p$, every path in
$\skyp_{(x,y);p}$ is a skyline point in $d$-dimensional space. Given an
integer $r$, an optimal $r$-partition $\mathcal{R}_{opt}$ is a
partition to minimize $\mathcal{D}(\mathcal{R})$. For every group $R_i$ in $\mathcal{R}_{opt}$, the
\emph{contour skyline point} $cp_i$ is the \lbop of the skyline
paths in $R_i$, the set of all $cp_i$ is called the
\emph{contour skyline set} of $\skyp_{(x,y);p}$, denoted as $\cs_{(x,y);p}$.
\end{definition}

The efficiency of the optimal path query can be improved by $\cs_{(x,y);p}$. We introduce it in Section \ref{subsec-query}. Next, we discuss how to compute the contour skyline points. This problem is to find the
optimal partition $\mathcal{R}_{opt}$ for all the skyline points in
$\skyp_{(x,y);p}$. In case of 2D space, we propose a dynamic
programming method to compute the optimal partition $\skyp_{(x,y);p}$. We prove this problem is NP-hard
in 3D or higher dimensional space. We give a 2-approximate algorithm and show there
is no $(2-\epsilon)$-approximate solution in the polynomial time.

\stitle{Case 1: (2D space)}: Assume that $\skyp_{(x,y);p}$ has been
already computed and let $m$ be the size of $\skyp_{(x,y);p}$. We use
$S=\{p_1,\cdots,p_m\}$ to denote the set of all skyline points in
$\skyp_{(x,y);p}$, where all $p_i$ in $S$ are sorted in ascending order
of their $x$-coordinates. We use $S_i$ to denote $\{p_1,
p_2,\cdots,p_i\}$. Specially, $S_0=\emptyset$. We also use a
notation $opt(i,t)$ to denote the optimal $t$-partition for $S_i$.
Obviously, the optimal $r$-partition $\mathcal{R}_{opt}$ for $S$ is
essentially $opt(m,r)$. Let $S_{j,i}$ be the point set $\{p_j,
\cdots, p_i\}$, where $0\leq j\leq i\leq m$. Then we have the
following recursive equation:
\begin{equation}\label{eq3}
\mathcal{D}(opt(i,t))=
\min\limits_{j=t-1}\limits^{i}\{\max\{\mathcal{D}(opt(j-1,t-1)),\mathcal{D}(S_{j,i})\}\}
\end{equation}
The meaning of Eq.~(\ref{eq3}) is that: without loss generality,
assume that the optimal $t$-partition of $S_i$ is $\{R_1,\cdots,R_t\}$,
where $R_t$ is the last group which consists of
$\{p_j,\cdots,p_i\}$. Then, $\{R_1,\cdots,R_{t-1}\}$ must be the
optimal $(t-1)$-partition for $S_{j-1}$. Let $j_{\min}$ be the value of $j$ minimizing Eq.~(\ref{eq3}), then we have
\begin{equation}\label{eq4}
\begin{aligned}
opt(i,t)=opt(j_{\min}-&1,t-1)\cup S_{j_{\min},i}\\
opt(i,1)&=S_i
\end{aligned}
\end{equation}

By Eq.~(\ref{eq3}) and Eq.~(\ref{eq4}), a dynamic
programming method can be utilized to compute the optimal $r$-partition for
$\skyp_{(x,y);p}$ in 2D space.

%

\comment{
\begin{algorithm}[t]
  \caption{\computecstwo($\skyp_x(u,v),r$)}
  \label{alg-dynamic}

{\small
\begin{tabbing}
{\bf Input:} \hspace{0.2cm}\= $\skyp_x(u,v)=\{p_1,\cdots,p_m\}$ and parameter $r$\\
{\bf Output:} \> Contour skyline set $\cs_x(u,v)$.
\end{tabbing}

\begin{algorithmic}[1]
 \FOR {$i =1$ to $m$}
   \STATE $\mathcal{D}(opt(i,1)) \leftarrow dist(p_1,p_i)$;
   \STATE $opt(i,1) \leftarrow S_i$
 \ENDFOR 
 \FOR {$t=2$ to $r$}
   \FOR {$i =1$ to $m$}
     \STATE $\mathcal{D}(opt(i,t)) \leftarrow \min\limits_{j=t-1}\limits^{i}\{\max\{\mathcal{D}(opt(j-1,t-1)),\mathcal{D}(S(j,i))\}\}$;
     \STATE $opt(i,t) \leftarrow opt(j_{\min}-1,t-1)\cup S(j,i)$;
   \ENDFOR
 \ENDFOR
 \STATE $\mathcal{R}_{opt} \leftarrow opt(m,r)$;
 \FOR {each $R_i \in \mathcal{R}_{opt}$}
    \STATE $cp_i \leftarrow \Phi(R_i)$; $\cs_x(u,v) \leftarrow \cs_x(u,v) \cup \{cp_i\}$;
  \ENDFOR 
 \STATE {\bf return} $\cs_x(u,v)$;
 \end{algorithmic}
}
\end{algorithm}
}

\stitle{Case2: (3D and the higher dimensional space)}: In 3D and the higher dimensional space , we prove the optimal
$r$-partition problem is NP-hard by reducing the $r$-split problem in 2D space, which is NP-hard, to this problem. Given a set of
points $\{p_1,\cdots,p_n\}$ in 2D space, the $r$-split problem is to
find a set of $r$ groups $\{B_1,\cdots,B_r\}$ that minimizes
\begin{equation}\label{eq5}
\max\limits_{1\leq x \leq r}\{\max\{dist(p_i,p_j)|p_i,p_j\in B_x\}\}
\end{equation}
This problem is similar to the $r$-partition problem for the skyline points,
but when the points in space are the skyline points, the complexity for the
$r$-split problem is unknown. We give Lemma \ref{lemma9} as follows:

\begin{lemma}\label{lemma9}
For dimensionality $d\geq 3$, the $r$-partition problem is NP-hard.
\end{lemma}

\begin{proof} Given a set of points $\{p_1,\cdots,p_n\}$ in 2D space,
we map each of them to a skyline point in 3D space. For a point $p_i$
with $x$-coordinate $p_i(x)$ and $y$-coordinate $p_i(y)$, it is mapped to a point $p'_i$ in 3D space with $x$,
$y$ and $z$-coordinates:
$p'_i(x)=-\frac{1}{\sqrt{2}}p_i(x)+\frac{1}{2}p_i(y)$,
$p'_i(y)=\frac{1}{\sqrt{2}}p_i(x)+\frac{1}{2}p_i(y)$, and
$p'_i(z)=-\frac{1}{\sqrt{2}}p_i(y)$. For any two points in 3D space
$p'_1$ and $p'_2$, if $p'_1(x)>p'_2(x)$ and $p'_1(y)>p'_2(y)$, then
$p'_1(z)<p'_2(z)$. It means each point in 3D space is a skyline
point. On the other hand, we also find $dist(p'_1,p'_2)=dist(p_1,p_2)$, where $dist(p_i,p_j)$ is the Euclidean distance between $p_i$ and $p_j$.
This reduction is in the polynomial time. If we can find the optimal
$r$-partition in the polynomial time, then we can solve $r$-split
problem in the polynomial time.

Given a set $S$ of points in 3D space, we can convert it to a $d$-dimensional point set $S'$ for any $d \geq 3$ easily. We assign $(d-3)$ zeros to all the other coordinates for any point in $S$. The optimal $r$-partition for $S'$ is obviously the optimal $r$-partition for $S$ in 3D space. It is in the polynomial time for the reduction from 3D space to the $d$-dimensional space. \eop
\end{proof}

We give a greedy algorithm for $r$-partition on
a given $\skyp_{(x,y);p}$ in a vertex subset $V_p$. The main idea is as follows: In the initialization phase, all the points are assigned to a group $R_1$. One of these points, denoted as $\head_{1}$, is selected as the ``base point'' of $R_1$. The selection of $\head_1$ is
arbitrary. During each iteration, some points in $R_1,\cdots,R_j$
are moved into a new group $R_{j+1}$. Also, one of these points will be
selected as the ``base point'' of the new group, i.e., $\head_{j+1}$. The construction of
the new group is accomplished by first finding a point $p_i$, in
one of the previous $j$ groups $\{R_1,\cdots,R_j\}$, whose distance
to the base point of group it belongs is maximal. Such a point will be
moved into the group $R_{j+1}$ and selected as the ``base point'' of $R_{j+1}$. A point
in any of the previous groups will be moved into group $R_{j+1}$ if
its distance to $p_i$ is not larger than the distance to the base point of
group it belongs to. With the $r$-partition, the $\cs_{(x,y);p}$ of $\skyp_{(x,y);p}$ can be computed easily according to the definition of the contour skyline set. 

\comment{
\begin{algorithm}[t]
  \caption{\computecs($\skyp_x(u,v),r$)}
  \label{alg-greedy}

{\small
\begin{tabbing}
{\bf Input:} \hspace{0.2cm}\= $\skyp_x(u,v)=\{p_1,\cdots,p_m\}$ and parameter $r$\\
{\bf Output:} \> Contour skyline set $\cs_x(u,v)$.
\end{tabbing}

\begin{algorithmic}[1]
 \STATE $R_1 \leftarrow S$; select a point $p \in S$ as  $\head_1$ randomly;
 \FOR {$i =1$ to $k-1$}
   \STATE let $p$ be the point whose distance to the base point of the group it belongs to is maximum in all groups;
   \STATE $R_{i+1} \leftarrow \{p\}$; $\head_{i+1} \leftarrow \{p\}$;
   \FOR {each $q \in \cup_{x=1}^{i}R_x$}
     \STATE let $R_j$ be the group that $p$ belongs to;
     \IF {$dist(q,\head_{i+1})\leq dist(q,\head_j)$}
       \STATE $R_{i+1} \leftarrow R_{i+1} \cup \{q\}$;
     \ENDIF
   \ENDFOR
   \STATE $\mathcal{R}_{opt} \leftarrow \mathcal{R}_{opt} \cup R_{i+1}$; 
 \ENDFOR
 \FOR {each $R_i \in \mathcal{R}_{opt}$}
   \STATE $cp_i \leftarrow \Phi(R_i)$; $\cs_x(u,v) \leftarrow \cs_x(u,v) \cup \{cp_i\}$;
 \ENDFOR 
 \STATE {\bf return} $\cs_x(u,v)$;
 \end{algorithmic}
}
\end{algorithm}
}

This algorithm is guaranteed as a $2$-approximate
solution because there is no $(2-\epsilon)$-approximate solution in the
polynomial time if $P \neq NP$, as analysis in \cite{DBLP:journals/tcs/Gonzalez85}.

In summary, for each $\skyp_{(x,y);p}$ in vertex subset $V_p$, we compute the
contour skyline set $\cs_{(x,y);p}$. We also maintain every
$\cs_{(x,y);p}$ in $I_p^S$.

\subsection{How to Partition Graph to K Vertex Subsets}\label{subsec-partition-k}

For optimal path problem in the multi-cost networks, the less number of edges among different vertex subsets results in the less number of entries and exits in the multi-cost 
network, and then the size of partition-based index becomes smaller. The objective of the partition is to make the edges dense in the same vertex subset and sparse among different vertex subsets. It is an optimal partition problem and has been well studied in the past couple of decades\cite{DBLP:conf/ipps/Abou-RjeiliK06,
DBLP:journals/pami/DhillonGK07, DBLP:conf/kdd/XuYFS07}. In this paper, we use the classic
multi-level graph partitioning algorithm, proposed by Metis et al. in
\cite{DBLP:conf/ipps/Abou-RjeiliK06}, to partition the networks in experiments.

%


%
\section{Query Processing} \label{main-query}
Given a multi-cost network $G(V,E,W)$, a starting vertex $v_s$ and an ending
vertex $v_e$, $V_s$ and $V_e$ are the vertex subsets including
$v_s$ and $v_e$ respectively. A shrunk graph $\bar{G}=(\bar{V},\bar{E})$ can be derived from partition-based index. $\bar{V}$
consists of three sets: (1) $V_s$; (2) $V_e$, and (3)
$\bigcup_{p\neq s,e}(V_{p}.entry \cup V_{p}.exit)$. The edges in $\bar{E}$ satisfy three following conditions: (1) $(v_i,v_j)\in \bar{E}$, iff $((v_i,v_j)\in E)
\wedge ((v_i,v_j\in V_s)\vee (v_i,v_j\in V_e))$; (2) $(v_i,v_j)\in \bar{E}$,
iff $((v_i,v_j)\in E) \wedge ((v_i\in V_p.exit)\wedge (v_j\in V_q.entry))$,
where $V_p\neq V_q$; and (3) $m$ edges $\{(v_i,v_j)^1,\cdots,(v_i,v_j)^m\}$ are constructed for any pair of entry $v_i$ and exit $v_j$ in $V_p$, where $V_p\neq V_s$ and $V_p \neq V_e$. Note that $m$ is the size of $\skyp_{(i,j);p}$. In case (3), every edge
$(v_i,v_j)^\alpha(1\leq \alpha\leq m)$ from $v_i$ to $v_j$ represents a skyline path
in $\skyp_{(i,j);p}$. The following theorem guarantees the optimal path problem on $G(V,E)$ is equivalent to that on  $\bar{G}(\bar{V},\bar{E})$.

\begin{theorem} \label{theorem-2}
Given a multi-cost graph $G(V,E)$, a starting vertex $v_s$ and an ending vertex $v_e$ on $G$, a shrunk graph $\bar{G}(\bar{V},\bar{E})$ regarding $v_s$ and $v_e$ can be constructed. Finding the optimal path from $v_s$ to $v_e$ in $G$ is equivalent to finding the optimal path from $v_s$ to $v_e$ in $\bar{G}$.
\end{theorem}

\begin{proof} 
First, we prove that an optimal path $p$ from $v_s$ to $v_e$ in $G$ is also an optimal path in $\bar{G}$. $p$ must be a path from $v_s$ to $v_e$ in $\bar{G}$, otherwise some part of $p$ can be dominated by a skyline path in a cluster. A new path can be constructed by using this skyline path instead of this part in $p$. By the monotonicity of the score function $f(\cdot)$, the score of new path is less than the score of $p$, which is contradict with that $p$ is the optimal path in $G$. Moreover, $p$ must be an optimal path from $v_s$ to $v_e$ in $\bar{G}$, otherwise there must exist another path $p'$ whose score is less that $p$ in $\bar{G}$. Obviously, $p'$ is also a path in $G$, thus it is contradict with that $p$ is the optimal path in $G$.  

Next, we prove that an optimal path $p$ in $\bar{G}$ is also an optimal path in $G$. Assume that there exist another path $p'$ whose score is less than $p$ in $G$, we consider two cases. First, $p'$ is also a path in $\bar{G}$, then $p$ is not the optimal path in $\bar{G}$ because $p'$'s score is less than $p$'s score. Second, $p'$ is not a path in $\bar{G}$, then $p'$ must be dominated by another path $p''$ in $\bar{G}$ and the score of $p''$ is less than the score of $p$ in $\bar{G}$. It is contradict with that $p$ is the optimal path in $\bar{G}$. \eop
\end{proof}

Based on Theorem \ref{theorem-2}, the optimal path from $v_s$ to $v_e$ on $G(V,E)$ is
equivalent to the optimal path on $\bar{G}(\bar{V},\bar{E})$. The
process of finding the optimal path includes two steps: (1) vertex-filtering;
and (2) query processing.

\begin{algorithm}[t]
  \caption{\vf($\bar{G}(\bar{V},\bar{E}),v_s,v_e,f(\cdot)$)}
  \label{alg4}

{\small
\begin{tabbing}
{\bf Input:} \hspace{0.2cm}\= $\bar{G}(\bar{V},\bar{E})$, the score function $f(\cdot)$, the starting vertex $v_s$ \\\> and the ending vertex $v_e$;\\
{\bf Output:} \> the optimal path $p^*_{s,e}$.
\end{tabbing}

\begin{algorithmic}[1]
 \STATE $\tau \leftarrow \min\{f(p_{s,e}^x|p_{s,e}^x\in
\mathcal{P}_{s,e}\}$;
 \FOR {each $v_i \in \bar{V}$}
    \IF {$\tau < f(\Phi_{s,i}+\Phi_{i,e})$}
       \STATE $\bar{V} \leftarrow \bar{V}-\{v_i\}$;
    \ENDIF
 \ENDFOR
 \STATE \findsp($\bar{G}(\bar{V}),v_s,v_e,f(\cdot)$)
 \STATE {\bf return} $p^*_{s,e}$, $\tau$;
 \end{algorithmic}
}
\end{algorithm}

\subsection{Vertex-Filtering} \label{subsec-vertexfilter}
We propose a vertex-filtering algorithm which can effectively filter
vertices from $\bar{G}(\bar{V},\bar{E})$. Given two vertices $v_i$
and $v_j$ in $\bar{G}$, $\Phi_{i,j}$ and $\mathcal{P}_{i,j}$ can be calculated by Algorithm \ref{alg2}. Obviously, $\tau=\min\{f(p_{s,e}^x)|p_{s,e}^x\in
\mathcal{P}_{s,e}\}$ is an upper bound of the score of the
optimal path from $v_s$ to $v_e$. If
$\mathcal{P}_{s,e}=\emptyset$, then there does not exist a path from $v_s$
to $v_e$ and algorithm immediately return $p^*{s,e}=\emptyset$. For any
$v_i$ in $\bar{G}$, if $\tau<f(\Phi_{s,i}+\Phi_{i,e})$, then $v_i$ can
be removed from $\bar{G}$. In the other words, the optimal path from $v_s$
to $v_e$ cannot pass through $v_i$. Theorem \ref{theorem1} guarantees
the correctness of the vertex filtering.

\begin{theorem} \label{theorem1}
Given a multi-cost graph $G(V,E)$, a score function $f(\cdot)$,
a starting vertex $v_s$ and an ending vertex $v_e$, a shrunk graph
$\bar{G}(\bar{V},\bar{E})$ can be constructed. $\mathcal{P}_{s,e}$
is the set of the single-one cost shortest paths from $v_s$ to $v_e$,
$\mathcal{P}_{s,e}\neq \emptyset$. $\tau$ is an upper bound of the
optimal path from $v_s$ to $v_e$,
$\tau=\min\{f(p_{s,e}^x)|p_{s,e}^x\in
\mathcal{P}_{s,e}\}$. For any vertex $v_i$ in $\bar{G}$, if
$\tau<f(\Phi_{s,i}+\Phi_{i,e})$, where $\Phi_{s,i}$ and $\Phi_{i,e}$
are the \lbop from $v_s$ to $v_i$ and the \lbop from $v_i$ to $v_e$
respectively, then the optimal path from $v_s$ to $v_e$ cannot travel
through $v_i$.
\end{theorem}

\begin{proof} 
We only need to prove that, for any path $p$ traveling through $v_i$, there
exists a path $p'$ without traveling through $v_i$, such that $f(p')<f(p)$. Obviously, $p$ consists of two
segments: (i) the sub-path $p_{s,i}$ from $v_s$ to $v_i$; and (ii) the
sub-path $p_{i,e}$ from $v_i$ to $v_e$. By the definition of the \lbop, we
have $\Phi_{s,i}\preccurlyeq p_{s,i}$ and $\Phi_{i,e} \preccurlyeq
p_{i,e}$. Thus, $\Phi_{s,i}+\Phi_{i,e}\preccurlyeq p$. By the
monotonicity of the score function $f(\cdot)$, $f(\Phi_{s,i}+\Phi_{i,e})\leq
f(p)$. Let $p'$ be the path in $\mathcal{P}_{s,e}$ whose score is
$\tau$, i.e., $f(p')=\tau$. Obviously, $p'$ is a path from $v_s$ to $v_e$ and it
does not travel through $v_i$, otherwise it is contradict with
$\tau<f(\Phi_{s,i}+\Phi_{i,e})$. Then we have
$f(p')<f(\Phi_{s,i}+\Phi_{i,e})\leq f(p)$.  \eop
\end{proof}

The vertex-filtering algorithm is shown in Algorithm \ref{alg4}. The algorithm
need to perform verification for every vertex in $\bar{V}$, then the
time complexity of the vertex-filtering algorithm is $O(\bar{V})$. $\bar{V}_f$ is 
the set of vertices that cannot be filtered in the
vertex-filtering step. Let $\bar{G}_f(\bar{V}_f,\bar{E}_f)$ be the
induced subgraph of $\bar{V}_f$ on $\bar{G}$. By Theorem
\ref{theorem1}, we only need to compute the optimal path from $v_s$ to
$v_e$ on $\bar{G}_f(\bar{V}_f,\bar{E}_f)$.

\subsection{Query Processing} \label{subsec-query}

We discuss the query processing for two cases: (1) score function is linear; and (2) score function is non-linear. 

For case (1), every pair of border vertex $v_i$ and entry $v_j$ can be calculated a score according to $\Phi_{i,j}$, and this score can be regarded as a lower bound of distance from one vertex subset to another. In addition, For every $SP_{(i,j);p}$ in Skyline-Path-Inner-Index $I_p^S$, the minimum score of the skyline path in $SP_{(i,j);p}$ is exactly the shortest distance from an entry $v_i$ to an exit $v_j$ in $V_p$. By calculating these score, the partition-based index becomes the G-Tree index proposed in \cite{DBLP:conf/cikm/ZhongLTZ13} and then the optimal path problem can be solved.

For case (2), the optimal path problem is NP-hard. A \emph{best-first} branch and bound search algorithm can be utilized to  compute the optimal path on
$\bar{G}_f(\bar{V}_f,\bar{E}_f)$ in the similar way as the algorithm proposed in \cite{DBLP:conf/cikm/YangYGL12}. Note that $\bar{G}$
is not a simple graph because there are several edges from an
entry $v_i$ to an exit $v_j$ in a vertex subset $V_p$. Given a graph $\bar{G}_f$, a starting vertex $v_s$ and an ending vertex $v_e$,
all the possible paths started from $v_s$ in $\bar{G}_f$ can be organized
in a search tree. Here, the root node represents the starting vertex set
$\{v_s\}$. Any non-root node
$C=\{v_s,(v_s,v_1),v_1,\cdots,(v_{l-1},v_{l}),v_l\}$ represents a
path started from $v_s$. $|C|$ is the number of vertices in $C$, i.e., $|C|=|\{v|v\in C\}|$.
For two different nodes $C$ and $C'$ in the search tree, $C$ is the
parent of $C'$ if they satisfy the following two conditions: (i)
$C\subset C'$ and $|C'|=|C|+1$; and (ii) $C'\setminus C$ is an \textit{edge-node} set $\{(v_i,v_j), v_j\}$, 
where $v_i$ and $v_j$ are the ending vertex of path
$C$ and $C'$ respectively. In each iteration, a node $C$ is
dequeued from the min-heap $H$. Algorithm extends $C$ by processing the
children of $C$. Assume that the ending vertex of $C$ is $v_i$. For each edge $(v_i,v_j)$ in 
$\bar{G}_f$, algorithm adds the edge-node set $\{(v_i,v_j),v_i\}$
into $C$ to get a child $C'$ of $C$. Note that there may exist several edges
from $v_i$ to $v_j$ when $v_i\in V_p.entry$ and $v \in V_p.exit$ and every
edge represents a skyline path from $v_i$ to $v_j$ in $G_p$. The similar pruning strategies in \cite{DBLP:conf/cikm/YangYGL12} can be used to decide whether
$C'$ can be pruned or not. If $C'$ cannot be pruned, it will be inserted into the
min-heap $H$. Algorithm terminates when $H$ is empty or $f(C)$ are not less that the minimum score of the path from $v_s$ to $v_e$ that has been searched 
for the top element $C$ in $H$. 

The contour skyline set can be used to improve the query efficiency. For an entry $v_i$ and an exit $v_j$ in a cluster
$V_p$, we use $\e_{i,j}=\{(v_i,v_j)^1,\cdots,(v_i,v_j)^m\}$ to denote the
multiple edges from $v_i$ to $v_j$. Each $(v_i,v_j)^\alpha\in \e_{i,j}$ represents
a skyline path in $\skyp_{(i,j);p}$. In each iteration, a node $C$ is to be expanded. Let
$v_i$ be the ending vertex of $C$. If $v_i$ is an entry of a cluster $V_p$($V_p\neq
V_s $ and $V_p\neq V_e$), then for each $v_j\in V_p.exit$, we do not
need to add every edge-node set $\{(v_i,v_j)^\alpha,v\}(1\leq \alpha \leq m)$ into $C$
to get a child $C'$ of $C$. Let $\cs_{(i,j);p}=\{cp_1,\cdots,cp_r\}$ be the
contour skyline set of $\skyp_{(i,j);p}$. Each $cp_x\in \cs_{(i,j);p}$
corresponds to a group $R_x$ of the skyline paths in
$\skyp_{(i,j);p}$ (recall $r$-partition), then $cp_x$ corresponds to a
group $\e_{i,j}^x$ of edges in $\e_{i,j}$, where $\e_{i,j}^x=\{(v_i,v_j)^{x_1},\cdots,(v_i,v_j)^{x_t}\}$, $\e_{i,j}^x \subset
\e_{i,j}$. Each $(v_i,v_j)^{x_\beta}\in \e_{i,j}^x$ represents a skyline path
in $R_x$. $cp_x$ can be considered as an edge from $v_i$ to $v_j$ and then $\{cp_x, v_j\}$ can be added 
into $C$ to get a virtual child $C'$ of $C$. $C'$
corresponds to a children group $C'_x=\{C'_{x_1},\cdots,C'_{x_t}\}$
of $C$, where each $C'_{x_\beta}(1\leq \beta \leq t)$ is a child of $C$,
$C'_{x_\beta}$ is obtained by adding the edge-node set $\{(v_i,v_j)^{x_\beta},v_j\}$ into
$C$. Because $cp_x$ is the \lbop of $R_x$, then $cp_x$ is the \lbop
of $\e_{i,j}^x$. Thus, we have $C'\prec C'_{x_\beta}$ for any $\beta,1\leq
\beta \leq t$. If the virtual node $C'$ can be pruned, then all
$C'_{x_\beta}$ in $C'_x$ can be pruned.

\section{Performance Study}\label{performance}
In this section, we test the partition-based index on six
real-life networks including road networks, social network, etc. All experiments were done on a 3.0 GHz
Intel Pentium Core i5 CPU PC with 32GB main memory, running on
Windows 7. All algorithms are implemented by Visual C++.

The details of real-life networks used in experiments are shown in Table \ref{table1}, where CAITN is the Chicago anonymized
internet trace network, CARN and EURN are two road networks of California and Eastern USA respectively, EuAll is an email communication network, Slashdot is a social network about technology related news, and HepPh is a citation network from the e-print arXiv.

\begin{table}
\begin{tabular}{c|c|c|c}
\hline Dataset & Category & Number of vertices & Number of edges \\
\hline CAITN  & IP network & 4,837  & 17,426 \\
\hline EuAll  & email network & 11,521 & 32,389 \\
\hline Slashdot & social network & 20,639 & 187,672 \\
\hline HepPh    & citation network & 34,546  & 421,578 \\
\hline CARN   & road network & 21,047 & 21,692 \\
\hline EURN   & road network & 3,598,623 & 4,354,029 \\
\hline
\end{tabular}
\vspace{0.2cm}
\caption{Dataset Characteristics} \label{table1}
\end{table}

For each network, we randomly assigned $d$ kinds of cost to every edge ($d\in
\{2,3,4,5\}$). We randomly generate 1,000 pairs of vertices and
query the optimal path for every pair . The reported querying
time is the average time on each dataset. The score
function is $f(w_1,\cdots,w_d)=\sum_{i=1}^d w_i^2$.

We compare our method with A* algorithm\cite{DBLP:journals/jacm/MandowP10}, genetic algorithm(GA)\cite{DBLP:conf/isat/Chomatek15} and LEXGO* algorithm\cite{DBLP:journals/eor/PulidoMP14}, which are three the state of the art heuristic algorithms for querying skyline paths over multi-cost graphs. Note that skyline paths essentially are a candidate set for an optimal path query, thus more time is necessary to seek out the optimal path from the skyline paths for these methods. The experimental results present the querying time of skyline path by these heuristic methods are always much larger than the optimal path by our method, even though the time are not counted in for finding an optimal one from all the skyline paths. We also compare our method with BF-Search in \cite{DBLP:conf/cikm/YangYGL12}, which uses a naive index to find the optimal path in the multi-cost networks under the non-linear functions.

\begin{table*}
\center {
\begin{tabular}{c|ccccc|ccccc}
\hline & \multicolumn{5}{|c}{$d=2$} & \multicolumn{5}{|c}{$d=3$}\\
Dataset  & A*  & GA & LEXGO* & BF-Search  & PB-Index  & A* & GA & LEXGO* & BF-Search & PB-Index \\
\hline
CAITN    & 28.37   & 8.76   & 10.13  & 0.0374 & 0.0041 & 47.26   & 12.42  & 16.52  & 0.0515 & 0.0071 \\
CARN     & 121.25  & 36.87  & 32.71  & 0.0733 & 0.0115 & 219.38  & 68.73  & 79.83  & 0.0851 & 0.0189 \\
EuAll    & 211.76  & 92.28  & 79.27  & 0.1471 & 0.0062 & 336.52  & 155.34 & 132.46 & 0.2019 & 0.0113 \\
Slashdot & 879.98  & 193.91 & 201.36 & 4.8139 & 0.0871 & 1127.62 & 316.77 & 289.71 & 6.2506 & 0.1027 \\
HepPh    & 1934.52 & 303.64 & 288.71 & 17.653 & 0.2194 & 3253.43 & 589.32 & 573.13 & 21.467 & 0.2938 \\
\hline
\end{tabular}
\vspace{0.2cm}
\caption{Online Querying Time in Second}\label{table2} }
\end{table*}

\begin{figure}[htbp]
\begin{center}
\subfigure[impact of $k$]{\label{fig7-1}
       \includegraphics[height=1.25in]{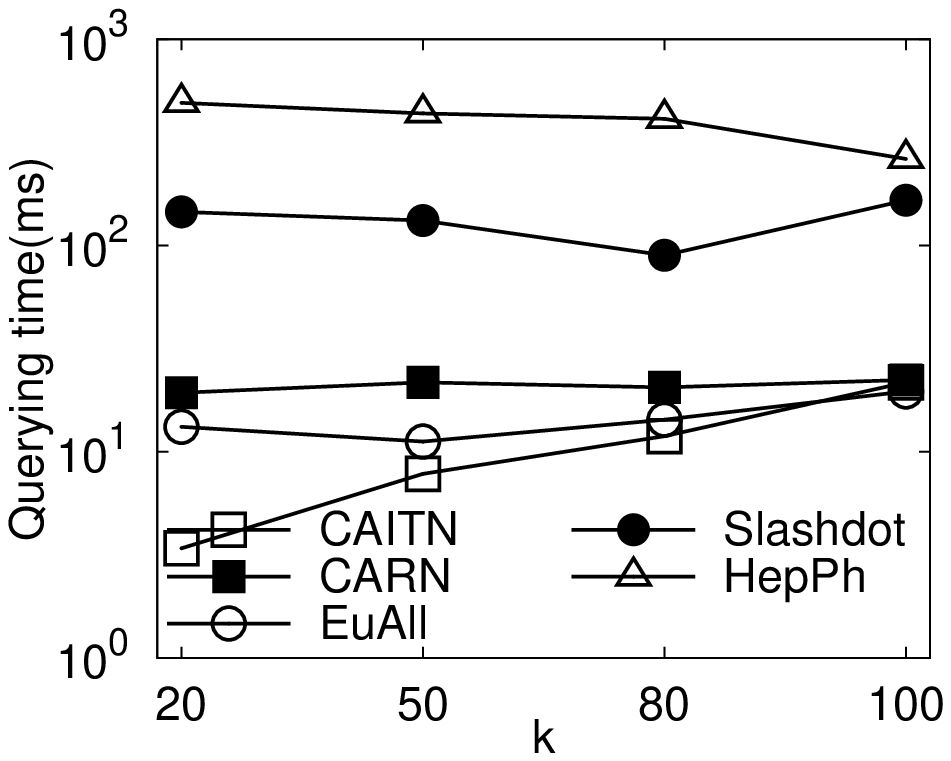}}\quad
\subfigure[impact of $r$]{\label{fig7-2}
       \includegraphics[height=1.25in]{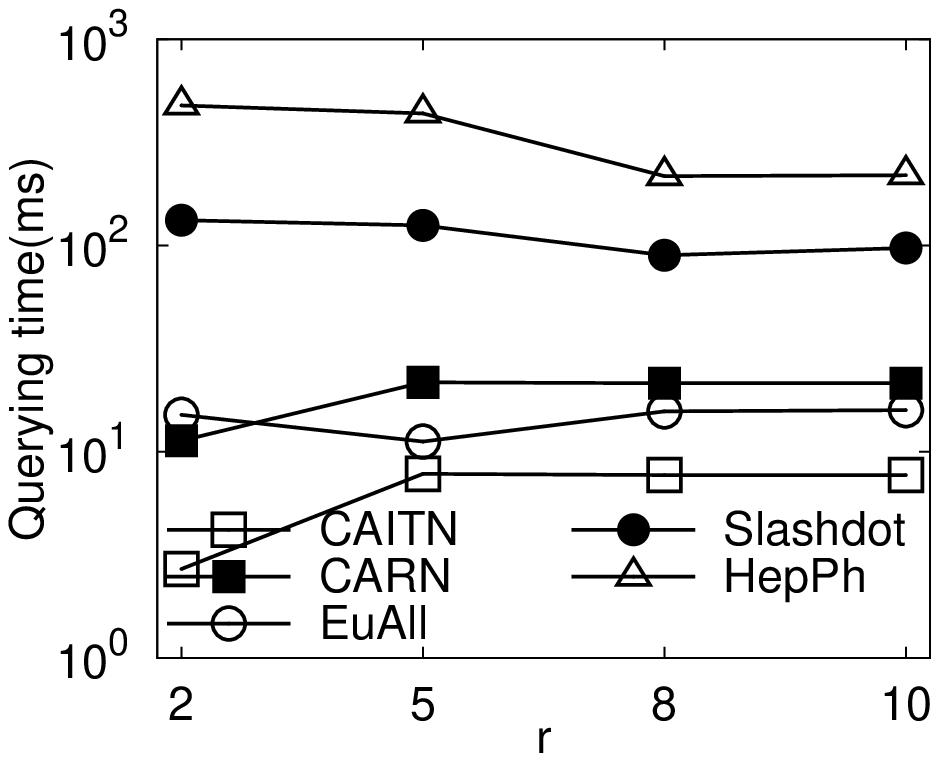}}
\end{center}
\caption{Impact of $k$ and $r$} \label{fig7}
\end{figure}

\stitle{Exp-1: Querying time}: As shown in Table \ref{table2}, we
investigate the querying time on five datasets by comparing
the partition-based index with A* algorithm, genetic algorithm, LEXGO* algorithm and BF-Search for $d=2$ and $d=3$. 
In this experiment, the number of vertex subsets
is $k=50$. For all networks, the querying time of the partition-based index are always in order of magnitude less than the others. 
The reason is that the partition-based index pre-computes the LBOP, skyline paths and contour skyline for any pair of entry and exit in every vertex subset and a large proportion
of the vertices are filtered in the vertex-filtering phase.

\stitle{Exp-2: Index size}: The index size is shown in Table
\ref{table3}. We compare the size of the partition-based index with the BF-Search for $d=2$ and $d=3$. A* algorithm, genetic algorithm and LEXGO* algorithm are not listed here because they do not use index.
The number $k$ is also $50$. We 
find the size of the the partition-based index are much smaller than
BF-Search. 
These results indicates the partition-based index is space efficient and it is more suitable for the large networks.

\begin{table}
\center {
\begin{tabular}{c|cc|cc}
\hline & \multicolumn{2}{|c}{$d=2$} & \multicolumn{2}{|c}{$d=3$}\\
Dataset & BF-Search  & PB-ndex  & BF-Search & PB-Index  \\
\hline
CAITN     & 115.99   & 6.21   & 203.78  & 13.52  \\
CARN      & 2600.68  & 93.85  & 4398.95 & 163.98 \\
EuAll     & 796.33   & 20.83  & 1333.86 & 39.23  \\
Slashdot  & 1746.39  & 47.21  & 3136.24 & 81.75  \\
HepPh     & 4124.96  & 138.74 & 6460.35 & 224.02 \\
\hline
\end{tabular}
\vspace{0.2cm}
\caption{Index Size in MB}\label{table3}}
\end{table}

\begin{table}
\center {
\begin{tabular}{c|c|c|c|c|c}
\hline  Dataset & $|\bar{V}|$ & $|\bar{E}|$ & $|\bar{V}_f|$ & $|\bar{E}_f|$ & $Avg.|\skyp_{(i,j);x}|$ \\
\hline
CAITN    & 746   & 19,132    & 368   &   9,560  & 11.17  \\
CARN     & 1,268 & 27,338    & 539   &  12,057  & 6.02   \\
Enron    & 1,073 & 29,418    & 471   &  13,715  & 14.78  \\
Slashdot & 1,782 & 293,877   & 936   &  198,429 & 43.16  \\
HepPh    & 3,832 & 1,718,753 & 1,297 &  646,396 & 55.31  \\
\hline
\end{tabular}
\vspace{0.2cm}
\caption{Impact of Vertex-Filtering}\label{table4} }
\end{table}

\stitle{Exp-3: Impact of vertex-filtering}: We investigate the
effectiveness of the vertex-filtering algorithm in Table \ref{table4}.
In this experiment, $k=50$ and $d=2$. From Table
\ref{table4}, we find the vertex-filtering algorithm can filter at
least $50\%$ vertices for each dataset. We find $|\bar{E}|$ may be
larger than $|E|$, where $|\bar{E}|$ and $|E|$ are the number of
vertices in the shrunk graph $\bar{G}$ and the original graph $G$
respectively. It is because that there are multiple edges between every 
pair of entry $v_i$ and exit $v_j$ in each $V_p$ ($V_p\neq V_s$ and
$V_p \neq V_e$) in $\bar{G}$. $Avg.|\skyp_{(i,j);p}|$ in Table \ref{table4} 
is the average number of the edges between any pair of entry $v_i$ and exit $v_j$ 
in the same vertex subset. In fact, for each pair of entry $v_i$ and exit $v_j$, $|\skyp_{(i,j);p}| \ll
|P_{(i,j);x}|$, where $P_{(i,j);x}$ is the number of all the possible paths
from $u$ to $v$ in $G_x$. Therefore, even though $|\bar{E}|>|E|$,
our algorithm on $\bar{G}$ are more efficient than that on $G$
because many paths from an entry to an exit have been
filtered by $\skyp_{(i,j);p}$. In addition, each edge $(v_i,v_j)^\alpha$ from
an entry $v_i$ to an exit $v_j$ in $\bar{G}$ represents a skyline path from $v_i$ to $v_j$. When algorithm expands a node $C$ whose ending vertex is
$v_i$, $C$'s children in $\bar{G}$ are more possible to be pruned than
that in $G$.

\stitle{Exp-4: Impact of $k$ and $r$}: We investigate the impact of the
number $k$ of the vertex subsets and the size $r$ of the contour skyline set. 
The experimental results are shown in Fig.~\ref{fig7}. For $k$, an appropriate
$k$ makes the number of the entries and the exits smaller in $\bar{G}$ and
thus the querying time is less. A larger or smaller $k$ will increase the
querying time. In Fig.~\ref{fig7-1}, we find the optimal $k$ are
distinct for the different datasets. For example, the optimal $k$ is 50
for Euall dataset but it is 80 for Slashdot dataset. For $r$, the skyline points in
a group are more proximity under a larger $r$ and then algorithm is
more effective to prune a virtual node $C'$ as the discussion in
section \ref{subsec-query}. On the other hand, a larger $r$ results
in the more contour skyline points and then the querying time increases. In two
extreme cases, when $r=1$, the only contour skyline point is the
\lbop of $\skyp_{(i,j);p}$, and when $r=|\skyp_{(i,j);p}|$, the contour
skyline set is exactly $\skyp_{(i,j);p}$. For these two cases, the
contour skyline set cannot work well. We find the optimal $r$ are also
distinct for the different datasets. The optimal $r$ is 5 for EuAll dataset and 
it is 8 for Slashdot and HepPh datasets.

\stitle{Exp-6. Scalability}: We evaluate the scalability of our method in Fig.\ref{fig8}. We investigate the querying time by varying the number of vertices from one million to three millions on EURN dataset for $d=2$ and $d=3$. For each graph, $k=10^{-3}n$, where $n$ is the number of the vertices in graph. We compare our method with BF-Search, GA algorithm and LEXGO* algorithm. The experimental results show our method are always in order of magnitude faster than others and it can perform efficiently even though the number of vertices is larger than three millions.  
It indicates our method are also suitable for large multi-cost graphs.

\begin{figure}[htbp]
\begin{center}
\subfigure[$d=2$]{\label{fig8-1}
       \includegraphics[height=1.15in]{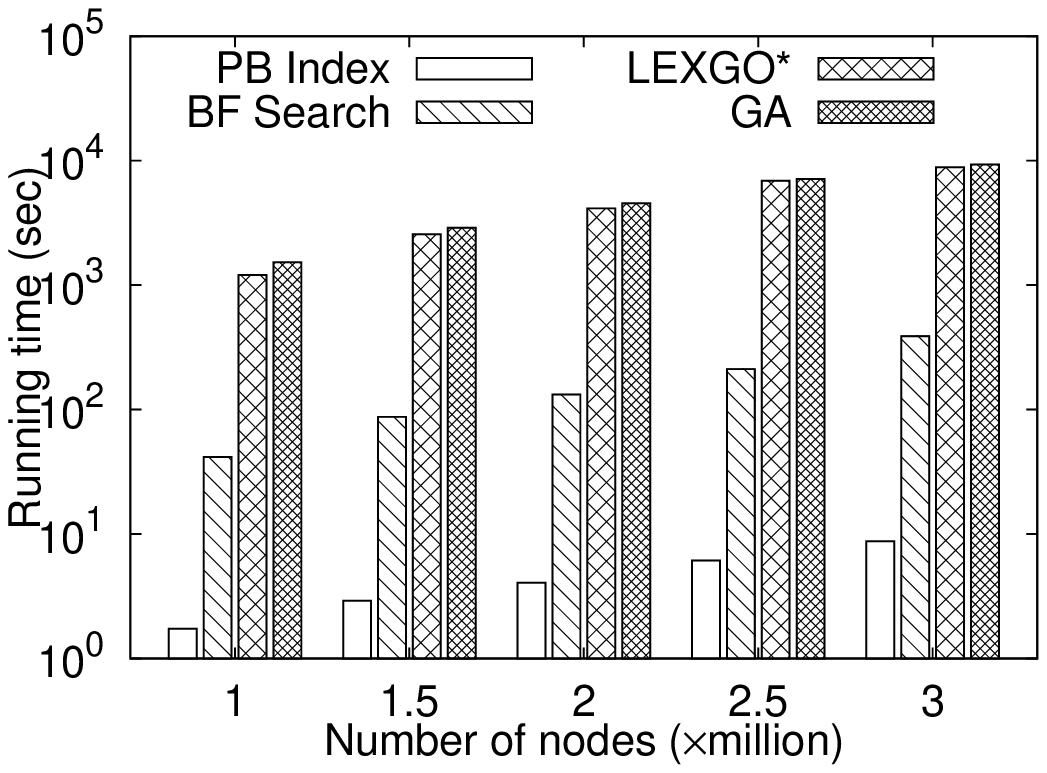}}
\subfigure[$d=3$]{\label{fig8-2}
       \includegraphics[height=1.15in]{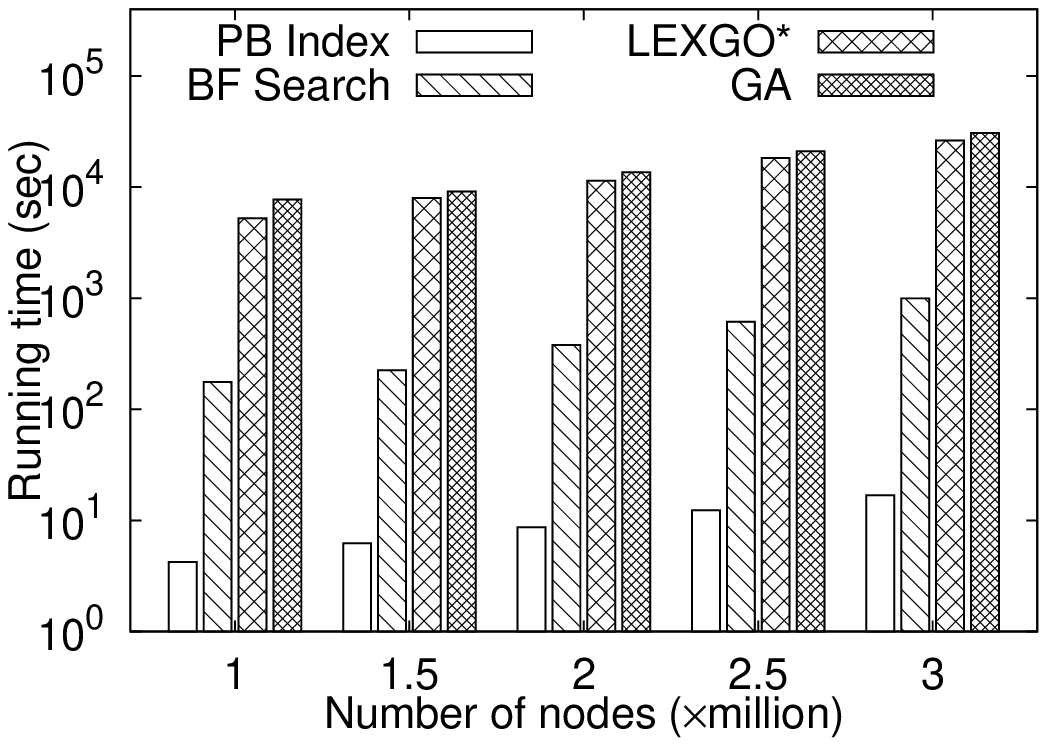}}
\end{center}
\caption{Adaptivity to large graphs} \label{fig8}
\end{figure}


%
\section{Related Work} \label{sec-related}

 

The existing works for the shortest 
path problem propose various index techniques to enhance the efficiency 
of the shortest path query for large graphs. \emph{The shortest
path quad tree} scheme is proposed in \cite{DBLP:conf/sigmod/SametSA08}, 
which pre-computes the shortest paths for every two vertices in a graph and 
organizes them by a quad tree. This method is not applicable for the optimal path 
problem in the multi-cost graphs. Because the score functions given by different users
may be different, the quad tree constructed according to one score
function cannot answer the optimal path query under the other functions.
Xiao et al. in \cite{DBLP:conf/edbt/XiaoWPWH09} proposes the concept
of the compact BFS-trees where the BFS-trees are compressed by
exploiting the symmetry property of the graphs. Wei et al. in
\cite{DBLP:conf/sigmod/Wei10} proposes a novel method named TEDI,
which utilizes the tree decomposition theory to build an index and
process the shortest path query. Cheng et al. in
\cite{DBLP:conf/sigmod/Cheng12} proposes a disk-based index for the single-source shortest path or distance queries. This
index is a tree-structured index constructed based on the concept of
vertex cover and it is I/O-efficient when the input graph is too
large to fit in main memory. Rice et al. in
\cite{DBLP:journals/pvldb/RiceT10} introduces a new shortest path
query type in which dynamic constraints may be placed on the
allowable set of edges that can appear on a valid shortest path.
They formalize this problem as a specific variant of formal language
constrained shortest path problems and then they propose the generalized shortest path queries in the following work\cite{DBLP:conf/icde/RiceT13}.  Zhu et al. in \cite{DBLP:conf/sigmod/ZhuMXLTZ13} presents AH index to narrow the gap between theory and practice. 
Landmark-based techniques have
been widely used to estimate the distance between two vertices in a
graph in many applications\cite{DBLP:conf/soda/GoldbergH05, DBLP:conf/icde/Qiao12,DBLP:conf/sigmod/AkibaIY13}. Goldberg et
al. in \cite{DBLP:conf/soda/GoldbergH05} choose some anchor vertices
called landmark and pre-computes for each vertex its graph distance
to all anchor vertices. A distance vector is created from these
distances. A lower bound derived from the distance vector can be
used by $A^*$ algorithm to guide the shortest path search.
Qiao et al. in \cite{DBLP:conf/icde/Qiao12} propose a query-dependent local
landmark scheme, which identifies a local landmark close to the
specific query nodes and provides a more accurate distance
estimation than the traditional global landmark approaches. The latest work\cite{DBLP:conf/sigmod/AkibaIY13} proposes a new exact method based on \textit{distance-aware 2-hop cover} for the distance queries. 
All the above methods utilize the following property in the shortest path: any sub-path of a shortest path is
also a shortest path. Therefore, they only need to maintain the
shortest paths among the vertices in the index and compute the shortest path by
concatenating the sub shortest paths in the index. However, in the
multi-cost graphs, this property does not hold. Therefore, these methods cannot solve the
optimal path problem in the multi-cost graphs.

In recent years, several works\cite{Martins1984236,Delling:2009,Mandow:2005,DBLP:conf/isat/Chomatek15,DBLP:journals/eor/PulidoMP14,DBLP:journals/jacm/MandowP10} study the multi-criteria 
shortest path (MCSP) problem on multi-cost graphs. Given a starting vertex and an ending vertex, it is to find all the 
skyline paths from the starting vertex to the ending vertex. Most existing works on MCSP are heuristic algorithm based on the following property: any sub-path 
of a skyline path is also a skyline path. To compute a skyline path $p$, these methods needs to expand all the skyline paths from the starting vertex to a vertex $v$ for every $v \in p$. The difference between MCSP and our problem is as follows. MCSP is to find all skyline paths but our problem is only to find one path that is the optimal under the score function. It is obvious that skyline paths is a candidate set of the optimal path. However, the time cost is too expensive to find an optimal path by exhausting all skyline paths. Moreover, these works does not develop any index technique to facilitate the skyline path querying. Mouratidis et al. in \cite{DBLP:conf/icde/MouratidisLY10}
studies the skyline queries and the top-k queries on the multi-cost
transportation networks. For any vertex $v$ in graph, all the distances
on the different dimensions between $v$ and the query point form the cost
vector of $v$. The definition of the cost vector in this work
is different with ours and the query results are points but
not paths. Therefore, the methods in this work cannot applied to the
optimal path problem in this paper.

\comment{
There are many works to study the skyline query based on shortest
path\cite{DBLP:conf/sigmod/PapadiasTFS03, DBLP:conf/edbt/ChenL08,
DBLP:conf/dasfaa/ZouCOZ10, DBLP:conf/icde/MouratidisLY10}. Papadias
et al. first introduces dynamic skyline problem in
\cite{DBLP:conf/sigmod/PapadiasTFS03}. Chen and Xiang proposes MSQ
algorithm for dynamic skyline problem in
\cite{DBLP:conf/edbt/ChenL08}, where the dimension function can be
any metric function. Zou et al. in \cite{DBLP:conf/dasfaa/ZouCOZ10}
studies dynamic skyline queries in a large graphs. Given some query
points, for any vertex in graph, the distances between it and query
point is considered as the cost vector of this vertex. The objective
of this work is to find skyline vertices based on this definition.
However, all these work assume that graph is a single-one cost
graph. Moreover, the complexity of their
problem is less than our problem.}

\section{Conclusion} \label{conc}
In this paper, we study the problem of finding the optimal route in
the multi-cost networks. We prove this problem is NP-hard and propose a novel partition-based index with contour skyline techniques.
We also propose a vertex-filtering algorithm to facilitate the query processing. We conduct extensive experiments and the experimental results validate
the efficiency of our method.


{
\small
\bibliographystyle{abbrv}
\bibliography{cikm}
}

\end{document}